\documentclass[c]{IEEEtran}
\usepackage{graphicx}
\usepackage{cite}
\usepackage[usenames,dvipsnames]{color}
\usepackage{amsmath}%
\usepackage{amsfonts}%
\usepackage{subfigure}
\usepackage{amssymb,bold-extra}%
\usepackage{graphicx}%
\usepackage{multicol}%
\usepackage{verbatim}%
\usepackage{enumerate}%------------------------
\usepackage{amsthm}
\usepackage{eucal}
\usepackage{epstopdf}
\newtheorem{theorem}{Theorem}
\newtheorem{lemma}[theorem]{Lemma}

\usepackage[center]{caption}

\begin{document}

\title{On the Diversity Gain Region of the Z-interference Channels}
\author{\large Mohamed S. Nafea$^*$, Karim G. Seddik$^\dag$, Mohammed Nafie$^*$, and Hesham El Gamal$^\S$\\ [.1in]
\normalsize  \begin{tabular}{c}
$^*$Wireless Intelligent Networks Center (WINC), Nile University, Cairo, Egypt. \\
$^\dag$Electronics Engineering Department, American University in Cairo, AUC Avenue, New Cairo, Egypt.\\
$^\S$Department of Electrical and Computer Engineering, Ohio State University, Columbus, USA.\\
\footnotesize Email: mohamed.nafea@nileu.edu.eg, kseddik@aucegypt.edu, mnafie@nileuniversity.edu.eg, helgamal@ece.osu.edu\\ \normalsize
\end{tabular}
\thanks{``This paper was made possible by a NPRP grant 09-1168-2-455 from the Qatar National Research Fund (a member of The Qatar Foundation). The statements made herein are solely the responsibility of the authors".}
\thanks{Mohammed Nafie is also affiliated with Faculty of Engineering, Cairo University.}\normalsize
 \vspace{-.4in}
}
 \maketitle

% ------------------------------ ABSTRACT -----------------------------
\begin{abstract}
In this work, we analyze the diversity gain region (DGR) of the single-antenna Rayleigh fading Z-Interference channel (ZIC). More specifically, we characterize the achievable DGR of the fixed-power split Han-Kobayashi (HK) approach under these assumptions. Our characterization comes in a closed form and demonstrates that the HK scheme with only a common message is a singular case, which achieves the best DGR among all HK schemes for certain multiplexing gains. Finally, we show that generalized time sharing, with variable rate and power assignments for the common and private messages, does not improve the achievable DGR.
\end{abstract}

\vspace{-.1in}
\section{Introduction}\label{Int}
The Z-Interference channel (ZIC) is the natural information theoretic model for many practical wireless communication systems. For example, the ``loud neighbor problem'' in femto-cells where a mobile station communicating with its long-range base station causes interference to the receiver of a short-range femto-cell \cite{LNP} can be accurately modeled as a ZIC. This paper analyzes the {\it{outage limited}} single antenna ZIC in the asymptotically large signal-to-noise ($\textsc{SNR}$) regime. Towards this end, we adopt the diversity-multiplexing tradeoff framework proposed by Zheng and Tse~\cite{TseDiv}.

The diversity-multiplexing tradeoff (DMT) for the interference
channel and the ZIC were studied in \cite{Bolcskei} and \cite{Doha}, respectively. The analysis of these earlier works assumed a {\em single} probability of system error which is given by the error probability of the worst user, and hence, failed to capture the tradeoff involved in the scenario where the two users require different Quality of Service (QoS) metrics. By considering the individual error performance metrics, our results extend this work and introduce the notion of the Diversity Gain Region (DGR) of the ZIC, which is the set of {\em simultaneously achievable} diversity gain pairs for the two users in the ZIC for a given multiplexing gain pair. Our main result is a complete characterization of the achievable DGR using the fixed power split Han-Kobayashi (HK) approach~\cite{Bolcskei}. This characterization is obtained in a closed form and sheds light on the structure of efficient communication strategies for the ZIC. For example, it is shown that, within the class of HK strategies, the special case of {\em common message only} (CMO) corresponds to a singular point, in the DGR, which is optimal in a certain range of multiplexing gains. In addition, the optimal choice of the splitting parameters in the other range of multiplexing gains is obtained. Finally, we show that generalized time sharing, with arbitrary power and rate allocations, does not improve the DGR of HK schemes.

The rest of this paper is organized as follows. In Section \ref{sysmod}, we describe our system model and notations. Section~\ref{DGRachive} obtains the achievable DGR as a function of splitting parameters and multiplexing gains and analyzes the two extreme special cases of Common Message Only (CMO) and Treating Interference As Noise (TIAN). In Section \ref{optimization}, we derive closed-form expressions for the DGR and optimal splitting parameters. In Section \ref{TS}, we prove that generalized time sharing of the HK scheme does not improve the achievable DGR. Section~\ref{Con} concludes the paper.

\section{System Model}\label{sysmod}
In this paper, we consider a two-user single-antenna communication system over a Rayleigh fading ZIC, i.e., transmitter 2 (TX2) causes interference to receiver 1 (RX1) but not vice versa as depicted in Fig. \ref{zmodel}. Each transmitter (TX$i,\;i =1,\;2)$ chooses a codeword $\textbf{x}_i \in\mathbb{C}^l$, $||\textbf{x}_{i}||^2\leq l$, from its codebook and transmits $\bold{\tilde{x}}_i=\sqrt{P_{i}}\textbf{x}_{i}$ according to its transmit power constraint $\bold{||\tilde{x}}_i||^2\leq lP_{i}$. We parameterize the attenuation of transmit signal $i$ at receiver $j$ (RX$j.\;j=1,\;2$) using the real-valued coefficients $\eta_{ij}>0$.

The input-output relations are given by
\begin{equation}
\begin{split}
&\bold{y}_1=\eta_{11} h_{11}\bold{\tilde{x}}_1+\eta_{21} h_{21}\bold{\tilde{x}}_2+\bold{n}_1\\ &\bold{y}_2=\eta_{22} h_{22}\bold{\tilde{x}}_2+\bold{n}_2,
\end{split}
\end{equation}
\noindent where $\bold{y}_i, \bold{n}_i\in\mathbb{C}^l$ denote the received codeword and the noise vector at RX$i$, respectively. The noise vectors are modeled as complex Gaussian random vectors with i.i.d. entries as $\bold{n}_1, \bold{n}_2\sim{\cal{CN}}(\bold{0},\bold{I}_l)$; the noise vectors are assumed to be temporally white. The channel gains $h_{11}$, $h_{21}$, and $h_{22}$ are i.i.d. complex Gaussian random variables with zero mean and unit variance. These gains are assumed to remain constant for a block of $l$ symbols and change randomly from one block to another. We also assume the channel gains to be known at the receiver but not at the transmitter. To simplify our results, we set $\eta_{11}^{2}P_{1}=\eta_{22}^{2}P_{2}=\textsc{SNR}$ and $\eta_{21}^{2}P_{2}=\textsc{SNR}^{\beta}$ with $\beta\geq0$. Thus, the input-output relations can be expressed as
\begin{equation}
\begin{split}
&\bold{y}_1=\sqrt{\textsc{SNR}} h_{11}\bold{x}_1+\sqrt{\textsc{SNR}^{\beta}} h_{21}\bold{x}_2+\bold{n}_1\\
&\bold{y}_2=\sqrt{\textsc{SNR}} h_{22}\bold{x}_2+\bold{n}_2.
\end{split}
\label{eq:zic_model}
\end{equation}

We express the exponential order of the channel gains as $|h_{ij}|^{2}=\textsc{SNR}^{-\gamma_{ij}}$. Notice that $|h_{ij}|^{2}$ is exponentially distributed with density $p_{|h|^{2}}(t)=e^{-t}$. By change of variables, it is a simple matter to show that the probability density function of $\gamma_{ij}$ in the high-$\textsc{SNR}$ is given by
\begin{equation}\small
p_{\gamma_{ij}}=\begin{cases}
0,&\mbox{\space for\space $\gamma_{ij}<0$}\\
\textsc{SNR}^{-\gamma_{ij}},&\mbox{\space for \space $\gamma_{ij}\geq0$}.
\end{cases}
\label{eq:gammapdf}
\end{equation}\normalsize

We consider an encoding scheme as a set of codebooks $\left\{C(\textsc{SNR})\right\}$ of block length $l$; one at each \textsc{SNR} level. An encoding scheme $\left\{C(\textsc{SNR})\right\}$ is said to achieve a multiplexing gain pair $(r_{1}, r_{2})$ and a diversity gain pair $(d_{1}, d_{2})$ for a ZIC if
\begin{equation}\small
\begin{split}
\lim_{\textsc{SNR}\rightarrow\infty}&\frac{R_{1}(\textsc{SNR})}{\log \textsc{SNR}}=r_{1},\;\;\;
\lim_{\textsc{SNR}\rightarrow\infty}\frac{R_{2}(\textsc{SNR})}{\log \textsc{SNR}}=r_{2}\\
\lim_{\textsc{SNR}\rightarrow\infty}&\frac{\log P_{e_1}(\textsc{SNR})}{\log \textsc{SNR}}=-d_{1},\;\;\;
\lim_{\textsc{SNR}\rightarrow\infty}\frac{\log P_{e_2}(\textsc{SNR})}{\log \textsc{SNR}}=-d_{2},
\end{split}
\end{equation}\normalsize
\noindent where $R_{1}(\textsc{SNR})$ and $R_{2}(\textsc{SNR})$ are the rates of user 1 and user 2, respectively, while $P_{e_1}(\textsc{SNR})$ and $P_{e_2}(\textsc{SNR})$ are their probabilities of error. Following the notations in \cite{TseDiv} and \cite{Tsediversity}, we define $P_e(\textsc{SNR})\doteq\textsc{SNR}^{-d}$ if $\lim_{\textsc{SNR}\rightarrow\infty}\frac{\log P_e(\textsc{SNR})}{\log\textsc{SNR}}=-d$, and $\dot{\leq}$, $\dot{\geq}$ are defined similarly.

\begin{figure}
	\centering
	\includegraphics[width=75mm, height = 40mm]{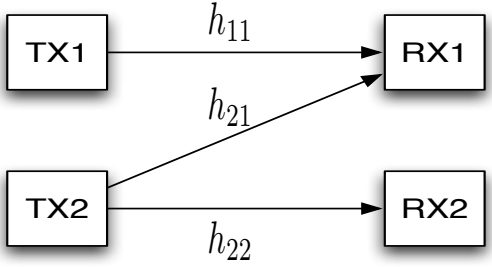}
	\caption{The ZIC model.}
	\label{zmodel}
\vspace{-.2in}
\end{figure}

According to the pdf of $\gamma_{ij}$ in (\ref{eq:gammapdf}) for i.i.d. random variables $\gamma_{11}$, $\gamma_{21}$, and $\gamma_{22}$, the probability $P_{out}$ that $\gamma_{11}$, $\gamma_{21}$, and $\gamma_{22}$  belong to a set $\mathcal{O}$ can be characterized by $P_{out}\doteq \textsc{SNR}^{-d_{out}}$ where $d_{out}= \underset{\gamma_{11},\gamma_{12},\gamma_{22}\in \mathcal{O}}\inf(\gamma_{11}+\gamma_{21}+\gamma_{22})$.

\section{DGR of Different Schemes in ZIC}\label{DGRachive}
In this section, we investigate the DGR for the two users in the ZIC given that they are operating at a multiplexing gain pair $(r_{1},r_{2})$ and under different schemes of encoding, transmission, and decoding. We consider the two-message fixed-power split HK approach applied at TX2 and two special cases of it where only a common or a private message is sent from TX2.

\subsection{DGR of the two-message fixed-power split HK approach}
We consider here the use of the two-message fixed-power split HK approach applied at TX2 \cite{Bolcskei,tuni} . More specifically, we define the private and common messages transmitted from TX2 with rates $S_{2}=s_{2} \log\textsc{SNR}$ and $T_{2}=t_{2}\log\textsc{SNR}$, respectively. Hence, $r_{2}=s_{2}+t_{2}$, $s_{2},t_{2}\geq 0$, and $0\leq r_{i} \leq 1$. Similar to \cite{Bolcskei}, we consider a joint Maximum Likelihood (ML) decoder at RX1 applied to the message of TX1 and the common message of TX2. At RX2, joint ML detection is carried out for both the private and common messages of TX2. For TX2, we parameterize the ratio of the average private power to the total average power as \cite{Doha, tuni}
\begin{equation}\small
\alpha=\frac{1}{1+\textsc{SNR}^b}\quad\in[0,1],\quad\quad{b}\in\mathbb{R}.
\end{equation}\normalsize
\noindent The received $\textsc{SNR}$s and Interference to Noise Ratios ($\textsc{INR}$s) on channels $h_{22},h_{21}$ are
\begin{equation}\small
\begin{split}
&\eta_{22}^{2}P_{2,{\rm{private}}}=\frac{\textsc{SNR}}{1+\textsc{SNR}^b},\;\; \eta_{22}^{2}P_{2,{\rm{common}}}=\frac{\textsc{SNR}^{1+b}}{1+\textsc{SNR}^b},\\
&\eta_{21}^{2}P_{2,{\rm{private}}}=\frac{\textsc{SNR}^{\beta}}{1+\textsc{SNR}^b},\;\; \eta_{21}^{2}P_{2,{\rm{common}}}=\frac{\textsc{SNR}^{\beta+b}}{1+\textsc{SNR}^b} .
\end{split}
\end{equation}\normalsize

We demonstrate that in the high-$\textsc{SNR}$ limit any choice of $b<0$ yields a zero diversity order. This contradicts the goal of improving the individual error exponents. Hence the region of optimization for the power splitting parameter $b$ is only over nonnegative values. In the high-$\textsc{SNR}$ limit, we can make the approximation that $1-\alpha\doteq1$.

The outage region $R_{\rm{HK,ZIC}}^{c}$ can be defined as
\begin{equation}\small
R_{\rm{HK,ZIC}}^{c}= \begin{cases}
&(\gamma_{11},\gamma_{21},\gamma_{22})\in \mathbb{R}_{+}^{3}:
R_{1}>\log\left(1+\frac{\textsc{SNR}^{1-\gamma_{11}}}{1+\frac{\textsc{SNR}^{\beta-\gamma_{21}}}{1+\textsc{SNR}^b}}\right)\\
&R_{1}+T_{2}>\log\left(1+\frac{\textsc{SNR}^{1-\gamma_{11}}+\textsc{SNR}^{\beta-\gamma_{21}}}{1+\frac{\textsc{SNR}^{\beta-\gamma_{21}}}{1+\textsc{SNR}^b}}\right)\\
&R_{2}>\log\left(1+\textsc{SNR}^{1-\gamma_{22}}\right)\\
&T_{2}>\log\left(1+\textsc{SNR}^{1-\gamma_{22}}\right)\\
&S_{2}>\log\left(1+\frac{\textsc{SNR}^{1-\gamma_{22}}}{1+\textsc{SNR}^b}\right).
\end{cases}
\label{eq:ORHK1}
\end{equation}\normalsize

We emphasize that, unlike a regular MAC region, no decoding error is declared at RX1 when erroneous decoding of the common message of TX2 occurs. Therefore, the corresponding outage event $T_2>\log\left(1+\frac{\textsc{SNR}^{\beta-\gamma_{21}}}{1+\frac{\textsc{SNR}^{\beta-\gamma_{21}}}{1+\textsc{SNR}^b}}\right)$ is not considered by RX1 \cite{Bolcskei}.

We can simplify this outage region by recognizing that the outage event $T_{2}>\log(1+\textsc{SNR}^{1-\gamma_{22}})$ is a subset of the outage event
$R_{2}>\log(1+\textsc{SNR}^{1-\gamma_{22}})$ and can be eliminated \cite{heshamHK}. The high-$\textsc{SNR}$ approximation for the outage region is given by
\begin{equation}\small
\begin{split}
R_{\rm{HK,ZIC}}^{c}=\begin{cases}
&(\gamma_{11},\gamma_{21},\gamma_{22})\in \mathbb{R}_{+}^{3}:\\
&r_1>\left[1-\gamma_{11}-\left[\beta-\gamma_{21}-b\right]^{+}\right]^{+}\\
&r_1+t_2>\Big[\max\left\{[1-\gamma_{11}]^+,[\beta-\gamma_{21}]^+\right\}\\
&\qquad\qquad\qquad-[\beta-\gamma_{21}-b]^+\Big]^+\\
&r_{2}>[1-\gamma_{22}]^+\\
&s_{2}>[1-\gamma_{22}-b]^+.
\end{cases}
\end{split}
\label{eq:ORHK}
\end{equation}\normalsize

\begin{theorem}
\label{HKdiversities}
The receiver diversities $d_{1,\rm{HK}}$ and $d_{2,\rm{HK}}$ of a two-user ZIC where the interfering transmitter TX2 uses the fixed-power split HK approach are given by
%For the two-user ZIC, when the transmitter of the second user, TX2, uses the two-message fixed-power split HK scheme, the RX1 diversity, $d_{1,\RM{HK}}$, and the RX2 diversity, $d_{2,\rm{HK}}$, can be expressed as functions of $t_2$ and $b$ as follows.
\begin{equation}\small
\begin{split}
&d_{1,\rm{HK}}(t_{2},b)=\min \left\{d_{11,\rm{HK}},d_{12,\rm{HK}}\right\},\qquad\text{where,}\\
&d_{11,\rm{HK}}=\left[1-r_{1}-\left[\beta-b\right]^+\right]^+\\
&d_{12,\rm{HK}}=\begin{cases}
\left[1-r_1-t_2\right]^{+}+\left[\beta-r_1-t_2\right]^+,\;\;b\geq r_1+t_2\\
\left[1-r_1-t_2-\left[\beta-b\right]^+\right]^{+},\;\; b<r_1+t_2,\\
\end{cases}\\
&\text{and,}\qquad d_{2,\rm{HK}}(t_2,b)=\min\{d_{21,\rm{HK}},d_{22,\rm{HK}}\},\\
&\text{where,}\;\;\;\; d_{21,\rm{HK}}=\left[1-r_2\right]^+\\
&\qquad\;\;\;\;\;\;d_{22,\rm{HK}}=\left[1-r_2-b+t_2\right]^+.
\end{split}
\label{eq:HKdiversities}
\end{equation}\normalsize
\end{theorem}

Using the outage region equations given in (\ref{eq:ORHK}), we can derive the individual diversities of the two-message fixed-power split HK scheme; the proof is omitted due to space limitations.

\subsection{DGR of the CMO special case}
Herein, we consider the special case of the HK approach where only a common message is sent from TX2; RX1 uses a joint ML decoder to jointly decode the intended message from TX1 and the interference message from TX2 whose rates are $R_1$ and $R_2$, respectively. Under this scheme, the high-$\textsc{SNR}$ approximation of the outage region can be written as
\begin{equation}\small
\begin{split}
R_{\rm{CMO,ZIC}}^c=\begin{cases}
&\left(\gamma_{11},\gamma_{21},\gamma_{22}\right)\in \mathbb{R}_{+}^{3}:\;\;r_1>[1-\gamma_{11}]^+\\
&r_1+r_2>\max\left\{[1-\gamma_{11}]^+,[\beta-\gamma_{21}]^+\right\}\\
&r_2>[1-\gamma_{22}]^+,
\end{cases}
\end{split}
\end{equation}\normalsize
\noindent which yields the following RX1 and RX2 diversities
\begin{equation}\small
\begin{split}
&d_{1,\rm{CMO}}=\min\left\{[1-r_1]^{+},[1-r_1-r_2]^++[\beta-r_1-r_2]^+\right\}\\
&d_{2,\rm{CMO}}=[1-r_2]^+ .
\end{split}
\label{eq:CMOdiversities}
\end{equation}\normalsize

The corresponding splitting parameters of the CMO special case are $b=\infty$ and $t_2=r_2$. It is noteworthy that we cannot obtain the resulting RX1 and RX2 diversities in the above equation directly by setting $b=\infty$ and $t_2=r_2$ in the expressions of RX1 and RX2 diversities for the HK approach given in (\ref{eq:HKdiversities}). We first have to remove the outage event
$S_2>\log\left(1+\frac{\textsc{SNR}^{1-\gamma_{22}}}{1+\textsc{SNR}^b}\right)$ from the outage region given in (\ref{eq:ORHK1}) where $S_2=0$
for this special case. Then, by removing the diversity expression corresponding to this outage event, $d_{22,\rm{HK}}$, and setting $b=\infty$ and $t_2=r_2$ in (\ref{eq:HKdiversities}), we get the results for $d_{1,\rm{CMO}}$ and $d_{2,\rm{CMO}}$
in (\ref{eq:CMOdiversities}). It is for this reason we demonstrate that the CMO scheme is a {\it{singular}} special case of the HK approach.

\subsection{DGR of the TIAN special case}
In this subsection, we consider the case where RX1 treats the interference from TX2 as additive white Gaussian noise with variance $|h_{21}|^{2}\textsc{SNR}^\beta$, i.e., TX2 sends a private message only. Under this scheme, the high-$\textsc{SNR}$ approximation of the outage region can be written as
\begin{equation}\small
\begin{split}
&R_{\rm{TIAN,ZIC}}^{c}=\\
&\begin{cases}
&\left(\gamma_{11},\gamma_{21},\gamma_{22}\right)\in \mathbb{R}_{+}^{2}:\;r_1>\left[1-\gamma_{11}-\left[\beta-\gamma_{21}\right]^+\right]^+\\
&r_2>\left[1-\gamma_{22}\right]^+.
\end{cases}
\end{split}
\end{equation}\normalsize

It is straight forward to show that when using the TIAN scheme, RX1 and RX2 diversities are given by
\begin{equation}\small
\begin{split}
&d_{1,\rm{TIAN}}=[1-r_1-\beta]^+,\qquad d_{2,\rm{TIAN}}=[1-r_2]^+.
\end{split}
\end{equation}\normalsize

The actual splitting parameters corresponding to the TIAN scheme are $b=-\infty$ and $t_2=0$. However, because of considering the high-$\textsc{SNR}$ approximation $\alpha\doteq\textsc{SNR}^{-b}$ in the outage region equations, the corresponding power splitting parameter $b$ becomes equal to zero. We can hence obtain the same expressions for $d_{1,\rm{TIAN}}$ and $d_{2,\rm{TIAN}}$ by setting $b=0$ and $t_2=0$ in the equations given in (\ref{eq:HKdiversities}). Therefore, the TIAN scheme is a direct special case of the HK approach.

\section{Optimization over the Rate and Power Splitting Parameters}\label{optimization}
The two-message fixed-power split HK approach and the two special cases of it, CMO and TIAN schemes, correspond to different power and rate splitting ratios between the private and common messages (at TX2), different encoding schemes, and different decoding algorithms. As previously mentioned, RX1 and RX2 diversities for the TIAN scheme can be obtained directly from the HK approach by setting $b=0$ and $t_2=0$. Therefore, when we consider optimizing the DGR over the splitting parameters for the HK approach, the TIAN scheme, contrary to the singular CMO counterpart, cannot improve the achievable DGR. Note that for the CMO scheme, RX2 diversity is similar to that of the single-user case. Thus, it cannot be increased using the HK scheme, which is not the case for RX1 diversity. For a given rate pair ($r_1,r_2$), obviously, either $d_{1,\rm{HK}}$ or $d_{1,\rm{CMO}}$ dominates. Therefore, the maximum simultaneously achievable RX1 and RX2 diversities of the general HK scheme are given by
\begin{equation}\small
\begin{split}
&d_{1}=\max \left\{d_{1,\rm{CMO}},d_{1,\rm{HK}}\right\}\\
&d_{2}=\begin{cases}
d_{2,\rm{CMO}},\quad\text{if}\;\; d_{1,\rm{CMO}}\geq d_{1,\rm{HK}}\\
d_{2,\rm{HK}},\quad\;\;\;\text{if}\;\; d_{1,\rm{CMO}}< d_{1,\rm{HK}},
\end{cases}\\
&\text{where,}\\
&d_{2,\rm{HK}}=\underset{\left\{(t_2,b):\;d_{1,\rm{HK}}(t_{2},b)=d_{1,\rm{HK}}\right\}}\max d_{2,\rm{HK}}(t_{2},b).
\end{split}
\end{equation}\normalsize

Our analysis here considers the low-level interference case, i.e., $\beta\leq1$, and a similar analysis can be carried out for the high-level interference. In order to characterize the achievable DGR of the general HK scheme, we have to specify the Multiplexing Gain Regions (MGRs) for which the condition $d_{1,\rm{HK}}>d_{1,\rm{CMO}}$ is satisfied for a non-empty set of the values of $t_2$ and $b$. This characterization is performed for asymmetric rates and stated in the following lemma.

\setcounter{theorem}{0}
\begin{lemma}
For the two-user ZIC, the following characterization specifies the MGRs and values of $t_2$ and $b$ for which $d_{1,\rm{HK}}$ is larger than $d_{1,\rm{CMO}}$.
\begin{enumerate}
\item For $\beta\geq r_1+2r_2$:\;
$d_{1,\rm{HK}}\leq d_{1,\rm{CMO}}$ \; for all values of $t_2$ and $b$.
\item For $r_1+r_2\leq\beta<r_1+2r_2$:\;
$d_{1,\rm{HK}}\geq d_{1,\rm{CMO}}$ \; for $b\geq \max\left\{r_1+t_2,2\beta-(r_1+2r_2)\right\}$ and all values of $t_2$.
\item For $\beta<r_1+r_2$:\;
$d_{1,\rm{HK}}\geq d_{1,\rm{CMO}}$ \; for $b\geq\beta-(\min\left\{r_2, 1-r_1\right\}-t_2)$ and all values of $t_2$.
\end{enumerate}
\end{lemma}

\begin{proof}
\noindent For each MGR, we compare $d_{1,\rm{HK}}$ for the different values of $t_2$ and $b$ to the corresponding $d_{1,\rm{CMO}}$ using the equations given in (\ref{eq:HKdiversities}) and (\ref{eq:CMOdiversities}). We will derive the stated result for the MGR $r_1+r_2\leq\beta<r_1+2r_2$. Using similar arguments, we can prove the results for the other MGRs.

For $r_1+r_2\leq\beta<r_1+2r_2$, RX1 diversity of the CMO scheme is given by $d_{1,\rm{CMO}}=\left[1+\beta-2(r_1+r_2)\right]^+$, while for the HK scheme, we perform the following analysis.

\textbf{For $r_1+t_2\leq b<\beta$:}\vspace{-0.05in}
\begin{equation*}\small
\begin{split}
&d_{12,\rm{HK}}=\left[1+\beta-2(r_1+t_2)\right]^+\geq\left[1+\beta-2(r_1+r_2)\right]^+.\\
&d_{11,\rm{HK}}=\left[1-r_1-(\beta-b)\right]^+\geq\left[1+\beta-2(r_1+r_2)\right]^+, \\
&\;\;\;\; \text{for}\;\;b\geq2\beta-(r_1+2r_2).\qquad\;\;\text{Thus,}\;\; d_{1,\rm{HK}}\geq d_{1,\rm{CMO}}.
\end{split}
\end{equation*}\normalsize

\textbf{ For $b<\beta,\;\;b<r_1+t_2$:}\vspace{-0.05in}
\begin{equation*}\small
\begin{split}
d_{12,\rm{HK}}&=\left[1-r_1-t_2-(\beta-b)\right]^+\leq\left[1-\beta\right]^+\\
&\leq\left[1+\beta-2(r_1+r_2)\right]^+,\\
\text{thus,}\;\; &d_{1,\rm{HK}}\leq d_{1,\rm{CMO}}.
\end{split}
\end{equation*}\normalsize

\textbf{For $b\geq\beta>r_1+t_2$:}\vspace{-0.05in}
\begin{equation*}\small
\begin{split}
&d_{11,\rm{HK}}=\left[1-r_1\right]^+\geq\left[1+\beta-2(r_1+r_2)\right]^+\\
&d_{12,\rm{HK}}=\left[1+\beta-2(r_1+t_2)\right]^+ \geq \left[1+\beta-2(r_1+r_2)\right]^+,\\
&\text{thus,}\;\;d_{1,\rm{HK}}\geq d_{1,\rm{CMO}}.
\end{split}
\end{equation*}\normalsize

\noindent Therefore, we have $d_{1,\rm{HK}}\geq d_{1,\rm{CMO}}$ for $b\geq\max\left\{r_1+t_2,2\beta-(r_1+2r_2)\right\}$ and for all values of $t_2$.
\end{proof}

After doing the previous characterization, we will derive closed-form expressions for the achievable tradeoff curve between $d_{1,\rm{HK}}$ and $d_{2,\rm{HK}}$ for the MGRs where $d_{1,\rm{HK}}$ can be larger than $d_{1,\rm{CMO}}$. This tradeoff is characterized through optimization of the simultaneously achievable RX1 and RX2 diversities of the HK scheme over the previously specified values of the splitting parameters $t_2$ and $b$. Our results are summarized in the following theorem.

\begin{theorem}
The tradeoff curve between $d_{1,\rm{HK}}$ and $d_{2,\rm{HK}}$ is characterized for the different MGRs as follows.\\\\
\textbf{For $r_1+r_2\leq\beta<r_1+2r_2$:}
\footnotesize
\begin{equation}
\begin{split}
&d_{2,\rm{HK}}=\\
&\begin{cases}
\left[1-r_1-\max\left\{r_2,2\left(\beta-r_1-r_2\right)\right\}\right]^+,\qquad\text{if}\;\;\;a_{11}\leq d_{1,\rm{HK}}\leq a_{12}\\
\text{using}\;\;b=d_{1,\rm{HK}}-1+r_1+\beta\;\;\text{and}\;\;t_2=d_{1,\rm{HK}}-1+\beta\\\\
\frac{1}{2}\left[5-\beta-4r_1-2r_2-3d_{1,\rm{HK}}\right]^+,\qquad\text{if}\;\;\;a_{12}\leq d_{1,\rm{HK}}\leq a_{13}\\
\text{using}\;\;b=d_{1,\rm{HK}}-1+r_1+\beta\;\;\text{and}\;\;t_2=\frac{1}{2}\left(1+\beta-2r_1-d_{1,\rm{HK}}\right).
\end{cases}
\end{split}
\end{equation}
\normalsize
\textbf{For $\beta<r_1+r_2$:}
\footnotesize
\begin{equation}
\begin{split}
&d_{2,\rm{HK}}=\\
&\begin{cases}
\left[1-\max\left\{r_2,\beta\right\}\right]^+,\qquad\text{if}\;\;\;a_{21}\leq d_{1,\rm{HK}}<a_{22}\\
\text{using}\;\;b=\beta\;\;\text{and}\;\;t_2=1-r_1-d_{1,\rm{HK}}\\\\
\left[2-r_1-r_2-\beta-d_{1,\rm{HK}}\right]^+,\qquad\text{if}\;\;\;a_{22}\leq d_{1,\rm{HK}}<a_{23}\\
\text{using}\;\;b=\beta\;\;\text{and}\;\;t_2=1-r_1-d_{1,\rm{HK}}\\\\
\left[1-\min\left\{r_1,\beta\right\}-r_2\right]^+,\qquad\text{if}\;\;a_{23}\leq d_{1,\rm{HK}}<a_{24}\\
\text{using}\;\;b=d_{1,\rm{HK}}-1+r_1+\beta\;\;\text{and}\;\;t_2=d_{1,\rm{HK}}-1+\beta\\\\
\frac{1}{2}\left[5-\beta-4r_1-2r_2-3d_{1,\rm{HK}}\right]^+,\;\;\;\text{if}\;\;a_{24}\leq d_{1,\rm{HK}}\leq a_{25}\\
\text{using}\;\;b=d_{1,\rm{HK}}-1+r_1+\beta\;\;\text{and}\;\;t_2=\frac{1}{2}\left(1+\beta-2r_1-d_{1,\rm{HK}}\right),
\end{cases}
\end{split}
\end{equation}
\normalsize
where,
\footnotesize
\begin{equation}
\begin{split}
&a_{11}=\left[1+\beta-2\left(r_1+r_2\right)\right]^+\\
&a_{12}=\max\left\{\left[1+\beta-2\left(r_1+r_2\right)\right]^+,\left[1-\frac{1}{3}\left(\beta+2r_1\right)\right]^+\right\}\\
&a_{13}=\left[1-r_1\right]^+,\qquad a_{21}=\left[1-r_1-r_2\right]^+\\
&a_{22}=\left[1-r_1-\min\left\{r_2,\beta\right\}\right]^+,\qquad a_{23}=\left[1-\max\left\{r_1,\beta\right\}\right]^+\\
&a_{24}=\left[1-\max\left\{r_1,\frac{1}{3}\left(\beta+2r_1\right)\right\}\right]^+,\qquad a_{25}=\left[1-r_1\right]^+.
\end{split}
\end{equation}
\normalsize
\end{theorem}

\begin{proof}
(Sketch)

\textbf{For $r_1+r_2\leq\beta<r_1+2r_2$:}\\
We perform the optimization of the achievable DGR for this MGR over the values of $b\geq\max\left\{r_1+t_2,2\beta-(r_1+2r_2)\right\}$ and for all values of $t_2$ where $t_2\in(0,r_2)$. Fig. \ref{fig2:proofa} specifies $d_{1,\rm{HK}}$ for each pair of these values of $t_2$ and $b$ and for the MGR $r_1+r_2\leq\beta<r_1+\frac{3}{2}r_2$. For each $d_{1,\rm{HK}}$ greater than $d_{1,\rm{CMO}}$, we choose the values of $t_2$ and $b$ which maximize $d_{2,\rm{HK}}$. In Fig. \ref{fig2:proofa}, all the points ($t_2,b$) on each horizontal line in regions 1 and 3 produce the same $d_{1,\rm{HK}}$. Also, in regions 2 and 4, all the points ($t_2,b$) on each vertical line produce the same $d_{1,\rm{HK}}$. Regions 1 and 2 give the same range of $d_{1,\rm{HK}}$ and choosing the points ($t_2,b$) on the line $b=r_1+t_2$ maximizes $d_{2,\rm{HK}}$ in these two regions. While for regions 3 and 4, the horizontal and vertical lines intersect in the line $b=2\beta-(r_1+2t_2)$ which gives the maximum $d_{2,\rm{HK}}$ for each $d_{1,\rm{HK}}$ achieved using any point ($t_2,b$) on any horizontal line in region 3 or any vertical line in region 4. Note that for this MGR, we always use $d_{2,\rm{HK}}=[1-r_2-(b-t)]^+$. Using the equations of $d_{1,\rm{HK}}$ and $d_{2,\rm{HK}}$ in those regions and the chosen values for $t_2$ and $b$, we can easily derive the expressions of the tradeoff curve for this MGR. When $r_1+\frac{3}{2}r_2\leq\beta<r_1+2r_2$, regions 1 and 2 in Fig. \ref{fig2:proofa} vanish, and thus the constant line on the tradeoff curve (See Fig. \ref{generalHK1}) will also vanish.

\textbf{For $\beta\leq r_1+r_2$:}\\
Using Fig. \ref{fig2:proofb} and similar arguments, we can derive the tradeoff curves for this MGR.
\end{proof}

\begin{figure}
\centering
\subfigure[$r_1+r_2<\beta<r_1+\frac{3}{2}r_2$]{
\includegraphics[width=76mm, height = 40mm]{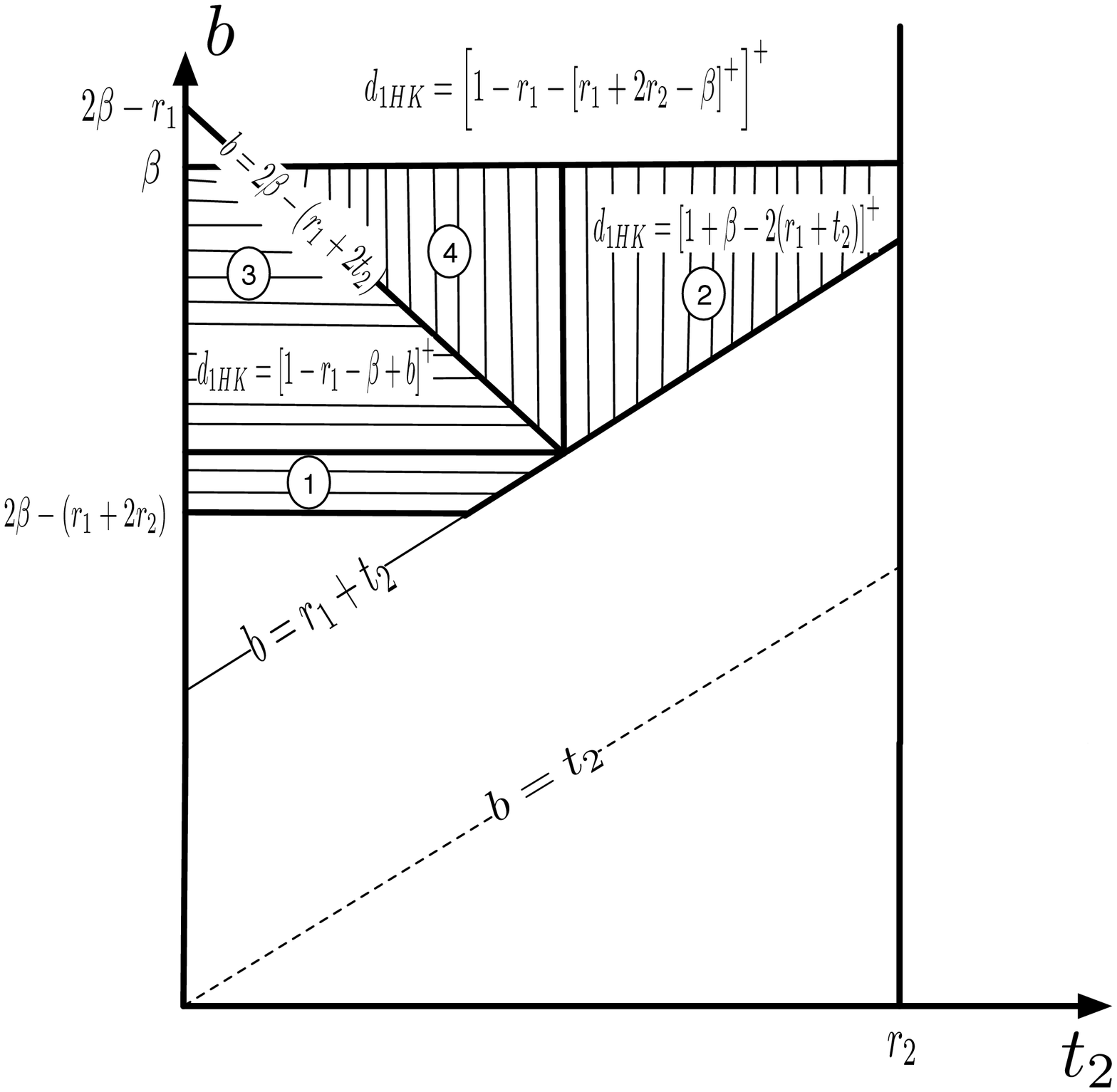}
\label{fig2:proofa}
 }
\subfigure[$r_1<\beta<r_2$]{
\includegraphics [width=80mm, height = 40mm]{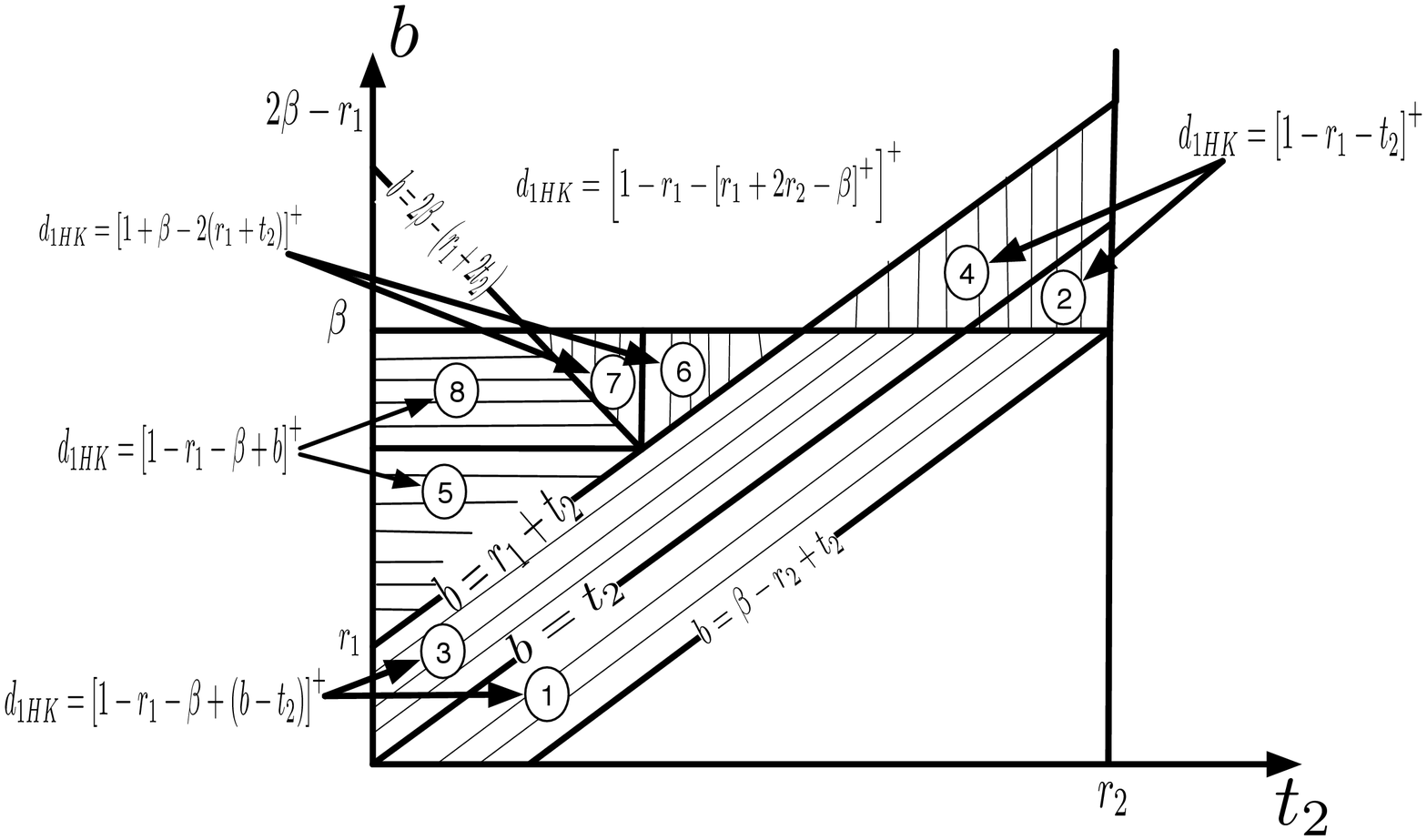}
\label{fig2:proofb}
 }
\caption{$d_{1,\rm{HK}}$ for the values of $t_2$ and $b$ which satisfy $d_{1,\rm{HK}}\geq d_{1,\rm{CMO}}$.}
\vspace{-.2in}
\end{figure}

The achievable tradeoff curves using the general HK scheme for the MGRs $r_1+r_2\leq\beta<r_1+\frac{3}{2}r_2$ and $r_1<\beta<r_2$ are shown in Fig. \ref{generalHK}.

\begin{figure}
\centering
\subfigure[$r_1+r_2\leq\beta<r_1+\frac{3}{2}r_2$]{
   \includegraphics[width=75mm, height = 30mm] {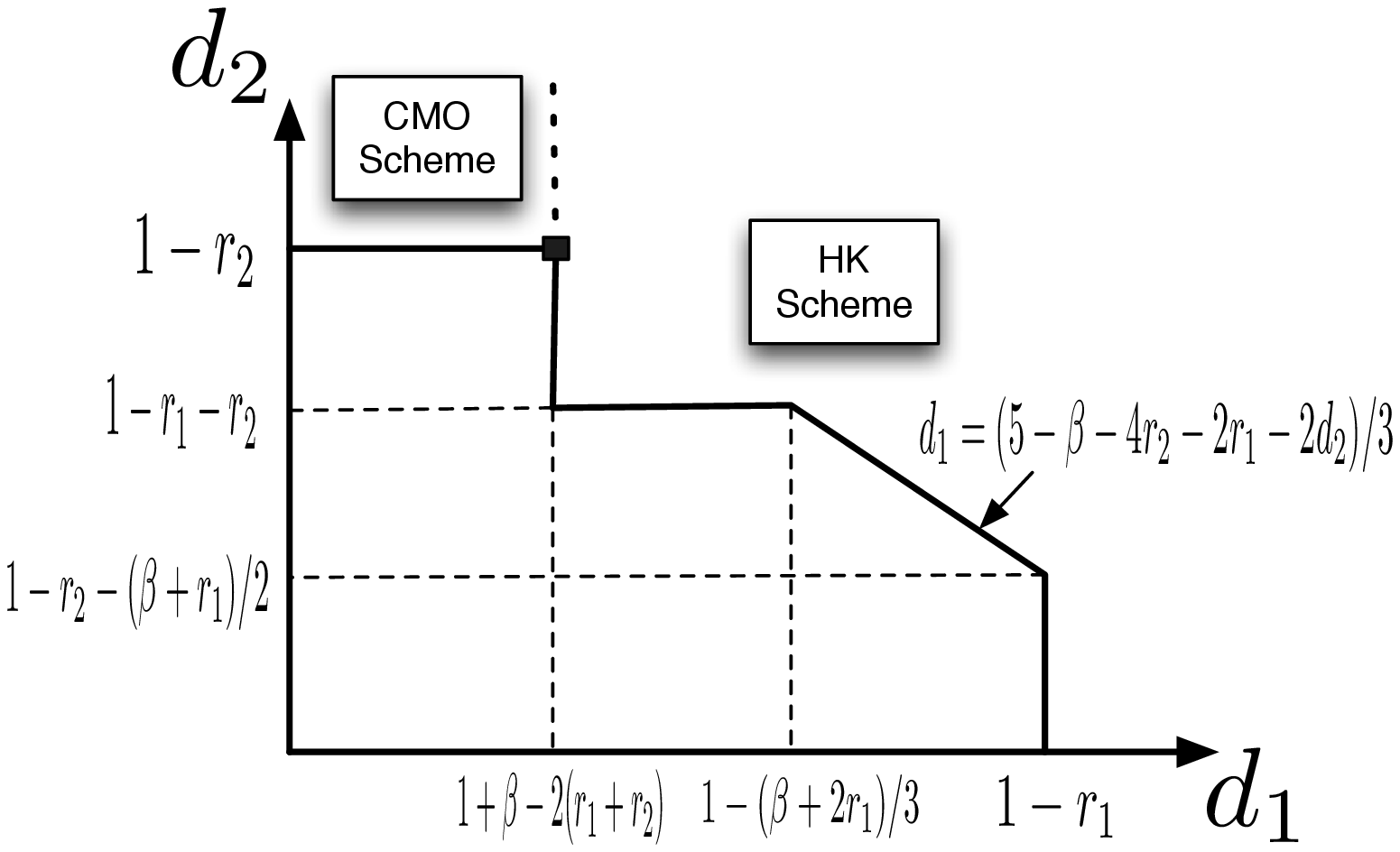}
   \label{generalHK1}
 }
  \subfigure[$r_1<\beta<r_2$]{
   \includegraphics[width=75mm, height = 30mm] {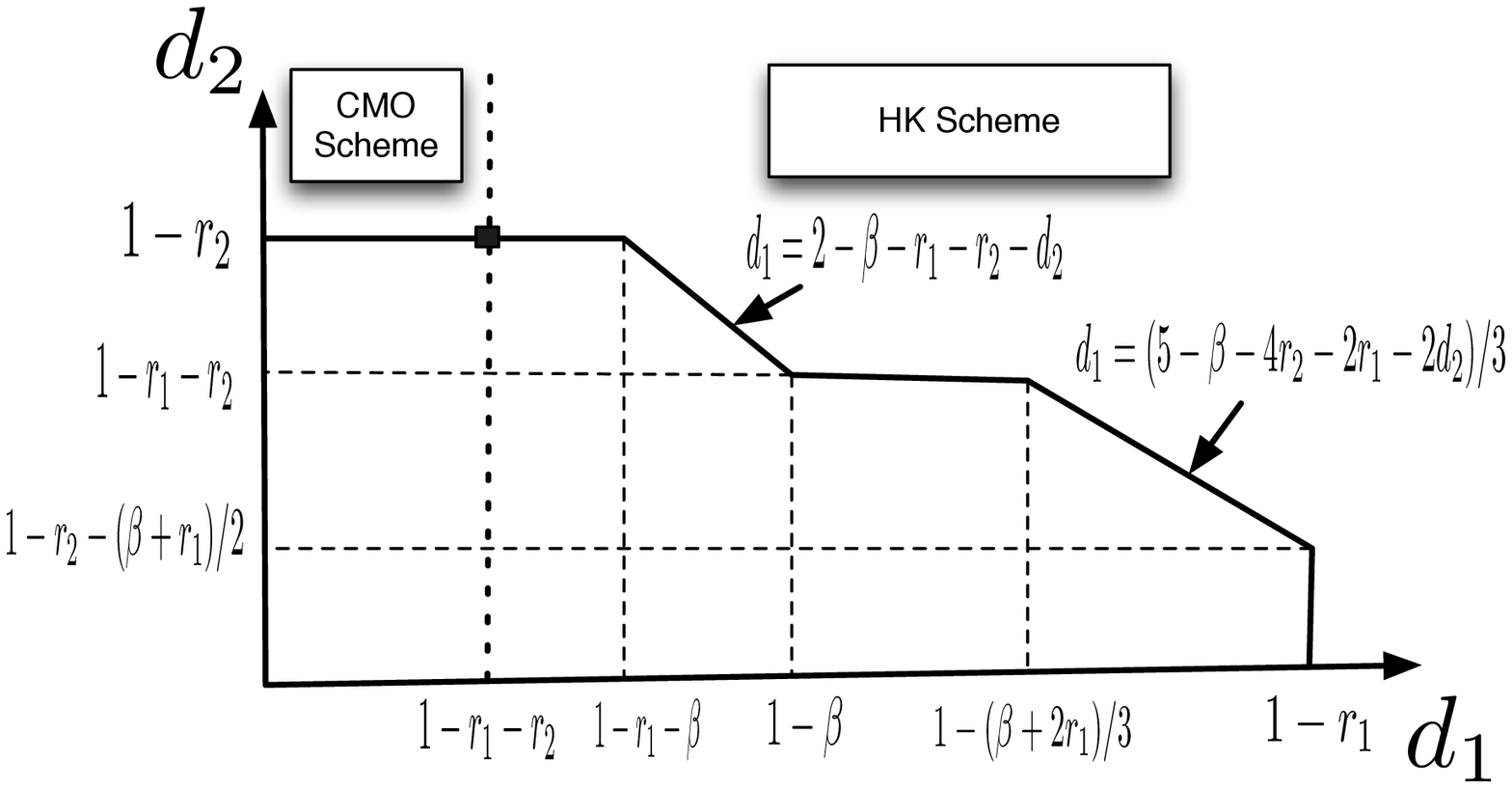}
   \label{generalHK2}
 }
\caption{The tradeoff curves of the general HK scheme for different MGRs.}
\label{generalHK}
\end{figure}

\section{Generalized Time Sharing of The HK Scheme}\label{TS}
In this section, we consider a generalized time sharing of the HK scheme. For TX2, we divide the block length into $L$ slots, $L\in\left\{1,2,3,\cdots\right\}$. For each slot, we consider different assignments for the rates and powers of the common and private messages via different encoding of the information bits. The constraint on the powers of the common and private messages in each time slot can be written as $P_{2,common}+P_{2,private}\leq P_2$. Thus, the case of having TX2 sending nothing in a particular slot can be considered. To ensure full transmission of TX2 data, the constraint on the rates is that the sum of the rates of the common and private messages in the whole $L$ slots has to be greater than or equal to $R_2$ with equality if and only if the pieces of information sent in the different slots are independent.

\begin{theorem}
Generalized time sharing of the general HK scheme does not improve the DGR achieved by fixed rate and power assignments for the common and private messages.
\end{theorem}

\begin{proof}
First, we will consider the case where the block length $l$ is divided into two slots; the first slot consists of $\lambda l$ symbols while the second one consists of $(1-\lambda)l$ symbols, where $0\leq\lambda\leq1$. We use the HK scheme with different rate and power assignments in each slot. It is obvious that when TX2 sends independent information bits through the two time slots, the overall individual error performance will be dominated by the worst performance in the  two time slots for both users; the individual error events for the common and private messages in the two slots still exist. Therefore, we concentrate on the case when TX2 sends correlated data in the two slots over the common and private messages but using independent codewords for each slot. In other words, some of the information is sent twice over the two time slots. To generalize the latter case, we have to consider the following three scenarios.
\begin{itemize}
\item Scenario 1: All of the information bits of TX2 are sent in each slot. In this scenario, TX2 encodes the information bits into three messages $m_1$, $m_2$, and $m_3$ and sends them over the common and private messages in the two time slots as shown in Fig. \ref{scenario1}. The message from TX1 is denoted by $m$.
\item Scenario 2: TX2 does not send all of the information bits in any slot but keeps some dependency between the information bits in the first and second slots. We can generally represent this scenario as shown in Fig. \ref{scenario2}. In this scenario, we consider that TX2 sends the same information message $m_1$ over the common messages in the two slots to minimize the amount of information to be decoded at RX1. On the other hand, the way we distribute the information messages over the common and private messages in the two slots will not affect the possible outage events at RX2 in the high-$\textsc{SNR}$ scale; however the power of the common message is much greater than that of the private message, both powers are very large. We also allow for some overlapping between the private messages in the two slots through $m_3$.
\item Scenario 3: TX2 introduces some redundancy through the common and private messages in each slot. This scenario is represented in Fig. \ref{scenario3}.
\end{itemize}

Let $T_{21}$, $S_{21}$, $T_{22}$, and $S_{22}$ be the rates of the common and private messages in the first and second slots, respectively, where $T_{21}=t_{21}\log \textsc{SNR}$ and $S_{21}$, $T_{22}$, and $S_{22}$ are similarly defined. Also, $R_1$ is the rate of the message $m$ from TX1.
For the first scenario, it is shown in the Appendix that the high-\textsc{SNR} approximation of the outage region at RX1 includes the following
outage events
\begin{equation}\small
\begin{split}
&r_1>[1-\gamma_{11}-[{\beta-\gamma_{21}-b_c}]^+]^+\\
&r_1+t_c>\left[\max\left\{[1-\gamma_{11}]^+,[\beta-\gamma_{21}]^+\right\}-[\beta-\gamma_{21}-b_c]^+\right]^+,
\end{split}
\label{eq:largesnrRX1}
\end{equation}\normalsize
where, $b_c=\lambda b_1+(1-\lambda)b_2$ and $t_c=\lambda t_{21}$.

The high-\textsc{SNR} approximation of the outage region at RX2 includes the following outage events
\begin{equation}\small
r_2>[1-\gamma_{22}]^+,\qquad r_2-t_c>[1-\gamma_{22}-b_c]^+.
\label{eq:largesnrRX2}
\end{equation}\normalsize

Note that these outage events are similar to the outage events of the HK scheme given in (\ref{eq:ORHK}),
but with the parameters $b_c$ and $t_c$ instead of $b$ and $t_2$. Since $b_c$ is a weighted sum of $b_1$ and $b_2$ while $t_c=\lambda t_{21}$, the DGR of the two-slots HK scheme with any choice of $b_1$, $b_2$, $t_{21}$, and $t_{22}$ cannot be increased over that of the fixed-power split HK scheme.

\begin{figure}
\centering
\subfigure[HK scheme, scenario 1]{
\includegraphics [scale=0.3]{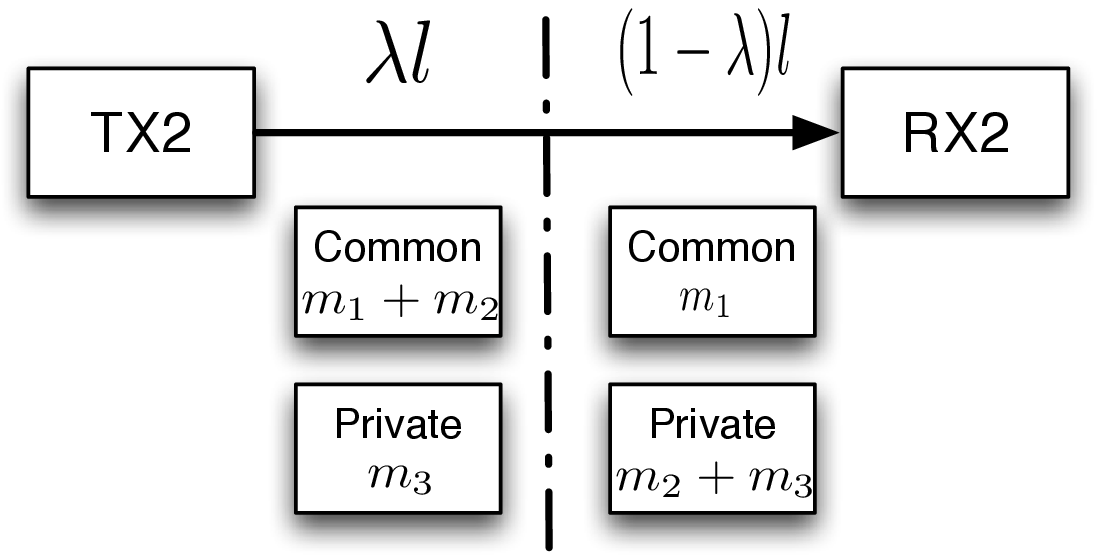}
\label{scenario1}
 }
\subfigure[HK scheme, scenario 2]{
\includegraphics [scale=0.3]{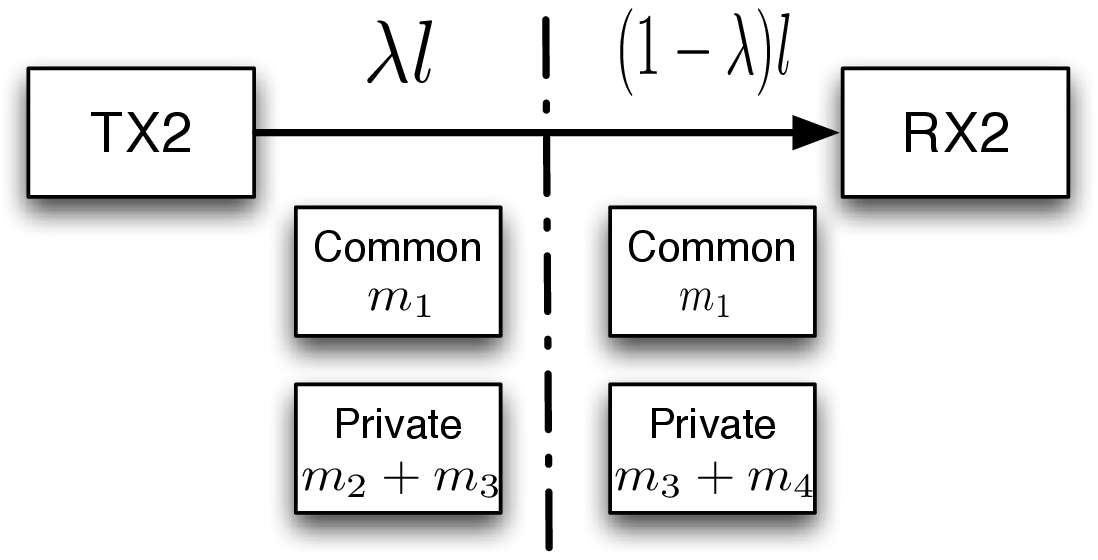}
\label{scenario2}
 }
\subfigure[HK scheme, scenario 3]{
\includegraphics [scale=0.3]{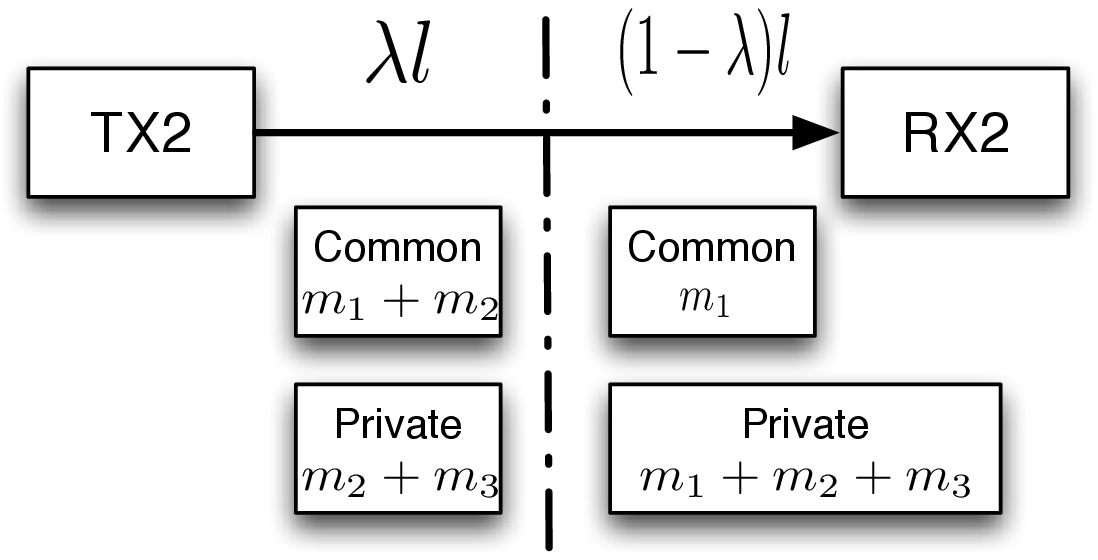}
\label{scenario3}
 }
 \subfigure[CMO scheme]{
\includegraphics [scale=0.3]{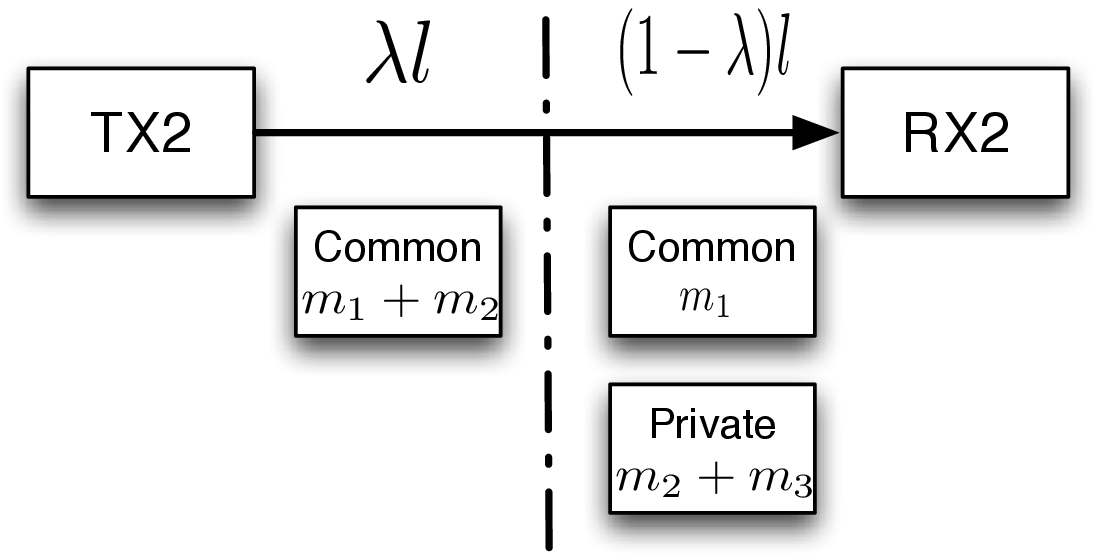}
\label{CMOscenario}
 }
\caption{Different scenarios for time sharing of the HK approach over two slots.}
\vspace{-.2in}
\end{figure}

It is straightforward to show that the high-$\textsc{SNR}$ approximation of the outage regions at RX1 and RX2 for the second and third scenarios include the same outage events as in (\ref{eq:largesnrRX1}) and (\ref{eq:largesnrRX2}). Using similar arguments to that for the first scenario, we can show that the second and third scenarios cannot also improve the achievable DGR.

Hence, we have shown that dividing the block length into two slots and using HK scheme with different rate and power assignment in each slot does not improve the achievable DGR. Using induction, we argue that dividing the block length into $L$ slots will not improve the achievable DGR.

Because of the singularity of the CMO scheme and in order to complete the proof, we need to show that dividing the block length into two slots and using the CMO scheme in the first slot while using the HK scheme in the second one does not improve the DGR achieved by using the general HK scheme. We now consider the scenario shown in Fig. \ref{CMOscenario}. It is straightforward to show that the high-$\textsc{SNR}$ approximation of the outage events at RX1 in this case can be written as
\begin{equation}\small
\begin{split}
&r_1>\lambda\left[1-\gamma_{11}\right]^++(1-\lambda)\left[1-\gamma_{11}-\left[\beta-\gamma_{21}-b\right]^+\right]^+\\
&r_1+\lambda t_{21}>\lambda \max\left\{\left[1-\gamma_{11}\right]^+,\left[\beta-\gamma_{21}\right]^+\right\}+(1-\lambda)\\
&\qquad\left[\max\left\{\left[1-\gamma_{11}\right]^+,\left[\beta-\gamma_{21}\right]^+\right\}-\left[\beta-\gamma_{21}-b\right]^+\right]^+\\
&r_1+\lambda t_{21}-(1-\lambda)t_{22}>\lambda \max\left\{\left[1-\gamma_{11}\right]^+,\left[\beta-\gamma_{21}\right]^+\right\}.
\end{split}
\label{eq:CMO-HK-RX1}
\end{equation}\normalsize
While the high-$\textsc{SNR}$ approximation of the outage events at RX2 can be reduced to
\begin{equation}\small
\begin{split}
&r_2>\left[1-\gamma_{22}\right]^+,\qquad r_2-\lambda t_{21}>(1-\lambda)\left[1-\gamma_{22}-b\right]^+\\
&r_2-(1-\lambda)t_{22}>\lambda\left[1-\gamma_{22}\right]^++(1-\lambda)\left[1-\gamma_{22}-b\right]^+,
\end{split}
\label{eq:CMO-HK-RX2}
\end{equation}\normalsize
where $b$ is the power splitting parameter in the second slot.

Using the outage events in (\ref{eq:CMO-HK-RX1}) and (\ref{eq:CMO-HK-RX2}), we can derive upper bounds on RX1 and RX2 diversities as follows.
\begin{equation}\footnotesize
\begin{split}
&d_1=\min\left\{d_{11},d_{12},d_{13}\right\}\leq\min\left\{d_{11},d_{12}\right\}\\
&\text{where,}\qquad d_{11}=\max\left\{\left[1-\frac{r_1}{\lambda}\right]^+,\left[1-r_1-(1-\lambda)\left[\beta-b\right]^+\right]^+\right\}\\
&d_{12}=\\
&\begin{cases}
\left[1-(r_1+\lambda t_{21})-(1-\lambda)\left[\beta-b\right]^+\right]^+,\\
\qquad\qquad\qquad\text{if}\;\;r_1+\lambda t_{21}\geq\lambda\beta+(1-\lambda)b,\;\;\;b\leq\beta\\
\left[1-\frac{(r_1+\lambda t_{21})-(1-\lambda)b}{\lambda}\right]^++\left[\beta-\frac{(r_1+\lambda t_{21})-(1-\lambda)b}{\lambda}\right]^+,\\
\qquad\qquad\qquad\text{if}\;\;b<r_1+\lambda t_{21}<\lambda\beta+(1-\lambda)b\\
\left[1-r_1-\lambda t_{21}\right]^++\left[\beta-r_1-\lambda t_{21}\right]^+,\;\;\;\text{if}\;\;r_1+\lambda t_{21}\leq b.\\
\end{cases}\\\\
&\text{And,}\qquad d_2=\min\left\{d_{21},d_{22},d_{23}\right\}\leq\min\left\{d_{21},d_{22}\right\}\\
&\text{where,}\qquad d_{21}=\left[1-r_2\right]^+,\qquad d_{22}=\left[1-\frac{r_2-\lambda t_{21}}{1-\lambda}-b\right]^+.
\end{split}
\end{equation}
\normalsize

Through optimization over the parameters $b$, $t_{21}$, and $\lambda$, we can prove that these upper bounds on RX1 and RX2 diversities do not improve the achievable DGR (proof omitted due to space limitations).
\end{proof}

\section{Conclusion}\label{Con}
In this paper, we derived closed-form expressions for an achievable four-dimensional tradeoff among the multiplexing gain pairs and the individual user diversity gain pairs for the ZIC. This region is obtained using the fixed power split HK scheme including its two special cases of CMO and TIAN. Interestingly, our characterization demonstrated the {\em singularity} of the CMO special case and its optimality in a certain range of multiplexing gains. Finally, our analysis revealed the {\em inability} of generalized time sharing to enlarge the DGR achieved with the HK approach.

%%-------------------------------------------------------------------------------------------------
 \appendix\label{Append1}

In this Appendix, we state the outage events at RX1 and RX2 for the first scenario shown in Fig. \ref{scenario1}. In this scenario, the rate of the message from TX1 $m$ is $R_1$ and it is encoded to the codeword $x_1$. Assume that the messages $m_1$, $m_2$, and $m_3$ with rates $R_{m_1}$, $R_{m_2}$, and $R_{m_3}$ are encoded to the codewords $x_{21}$, $x_{22}$, and $x_{23}$, respectively. Considering that RX1 is a MAC receiver of the three messages $m$, $m_1$, and $m_2$, the following outage events are included in the outage region at RX1.
\begin{equation}\small
R_1>I\left(x_1;y_1|x_{21},x_{22}\right),
\;\; R_1+R_{m_1}+R_{m_2}>I\left(x_1,x_{21},x_{22};y_1\right),
\end{equation}\normalsize
where, $R_{m_1}+R_{m_2}=\lambda T_{21}$. These outage events can be written as
\begin{equation}\footnotesize
\begin{split}
&R_1\geq\lambda\log\left(1+\frac{\textsc{SNR}^{1-\gamma_{11}}}{1+\frac{\textsc{SNR}^{\beta-\gamma_{21}}}{1+\textsc{SNR}^{b_1}}}\right)+(1-\lambda)\log\left(1+\frac{\textsc{SNR}^{1-\gamma_{11}}}{1+\frac{\textsc{SNR}^{\beta-\gamma_{21}}}{1+\textsc{SNR}^{b_2}}}\right)\\
&R_1+\lambda
T_{21}\geq\lambda\log\left(1+\frac{\textsc{SNR}^{1-\gamma_{11}}+\textsc{SNR}^{\beta-\gamma_{21}}}{1+\frac{\textsc{SNR}^{\beta-\gamma_{21}}}{1+\textsc{SNR}^{b_1}}}\right)\\
&\qquad\qquad\qquad\qquad\qquad+(1-\lambda)\log\left(1+\frac{\textsc{SNR}^{1-\gamma_{11}}+\textsc{SNR}^{\beta-\gamma_{21}}}{1+\frac{\textsc{SNR}^{\beta-\gamma_{21}}}{1+\textsc{SNR}^{b_2}}}\right).
\end{split}
\end{equation}\normalsize
Thus, the high-$\textsc{SNR}$ approximation of these outage events can be written as given in (\ref{eq:largesnrRX1}).
\begin{comment}
\begin{equation}
\begin{split}
&R_1+\lambda
T_{21}\geq\lambda\log\left(1+\frac{\textsc{SNR}^{1-\gamma_{11}}+\textsc{SNR}^{\beta-\gamma_{21}}}{1+\frac{\textsc{SNR}^{\beta-\gamma_{21}}}{1+\textsc{SNR}^{b_1}}}\right)\\
&\qquad+(1-\lambda)\log\left(1+\frac{\textsc{SNR}^{1-\gamma_{11}}+\textsc{SNR}^{\beta-\gamma_{21}}}{1+\frac{\textsc{SNR}^{\beta-\gamma_{21}}}{1+\textsc{SNR}^{b_2}}}\right)\\
&R_1+\lambda T_{21}-(1-\lambda)T_{22}\geq\lambda\log\left(1+\frac{\textsc{SNR}^{1-\gamma_{11}}+\textsc{SNR}^{\beta-\gamma_{21}}}{1+\frac{\textsc{SNR}^{\beta-\gamma_{21}}}{1+\textsc{SNR}^{b_1}}}\right)\\
&R_1+(1-\lambda) T_{22}\geq\lambda\log\left(1+\frac{\textsc{SNR}^{1-\gamma_{11}}+\textsc{SNR}^{\beta-\gamma_{21}}}{1+\frac{\textsc{SNR}^{\beta-\gamma_{21}}}{1+\textsc{SNR}^{b_1}}}\right)\\
&\qquad+(1-\lambda)\log\left(1+\frac{\textsc{SNR}^{1-\gamma_{11}}+\textsc{SNR}^{\beta-\gamma_{21}}}{1+\frac{\textsc{SNR}^{\beta-\gamma_{21}}}{1+\textsc{SNR}^{b_2}}}\right).
\end{split}
\end{equation}

Recall that no decoding error is declared at RX1 when decoding either $m_1$ or $m_2$ erroneously, Thus, their
corresponding individual and joint outage events are eliminated.  Without loss of generality, we assume that
$\lambda T_{21}\geq(1-\lambda)T_{22}$. Hence, the last outage event can be eliminated.
\end{comment}

We can also consider that RX2 is a MAC receiver of the messages $m_1$, $m_2$, and $m_3$. Thus, the following outage events are included in the outage region at RX2.
\begin{equation}\footnotesize
R_{m_3}>I\left(x_{23};y_2|x_{21},x_{22}\right),
\;\; R_{m_1}+R_{m_2}+R_{m_3}>I\left(x_{21},x_{22},x_{23};y_2\right),
\end{equation}\normalsize
where, $R_{m_3}=R_2-\lambda T_{21}$ and $R_{m_1}+R_{m_2}+R_{m_3}=R_2$. We can write these outage events at RX2 as follows.
\begin{equation}\footnotesize
\begin{split}
&R_2-\lambda T_{21}>\lambda\log\left(1+\frac{\textsc{SNR}^{1-\gamma_{22}}}{1+\textsc{SNR}^{b_1}}\right)+(1-\lambda)\log\left(1+\frac{\textsc{SNR}^{1-\gamma_{22}}}{1+\textsc{SNR}^{b_2}}\right)\\
&R_2>\log\left(1+\textsc{SNR}^{1-\gamma_{22}}\right),
\end{split}
\end{equation}\normalsize
which can be reduced to the outage events given in (\ref{eq:largesnrRX2}) for the high-$\textsc{SNR}$ approximation.
%\begin{equation}
%\begin{split}
\begin{comment}
\begin{enumerate}
\item $(1-\lambda)T_{22}>\log\left(1+\textsc{SNR}^{1-\gamma_{22}}\right)$\\
\item $\lambda T_{21}-(1-\lambda)T_{22}>\lambda\log\left(1+\textsc{SNR}^{1-\gamma_{22}}\right)\\
\;\;\;+(1-\lambda)\log\left(1+\frac{\textsc{SNR}^{1-\gamma_{22}}}{1+\textsc{SNR}^{b_2}}\right)$\\
\item $R_2-\lambda T_{21}>\lambda\log\left(1+\frac{\textsc{SNR}^{1-\gamma_{22}}}{1+\textsc{SNR}^{b_1}}\right)\\
\;\;\;+(1-\lambda)\log\left(1+\frac{\textsc{SNR}^{1-\gamma_{22}}}{1+\textsc{SNR}^{b_2}}\right)$\\
\item $\lambda T_{21}>\log\left(1+\textsc{SNR}^{1-\gamma_{22}}\right)$\\
\item $\lambda S_{21}-(1-\lambda)T_{22}>\log\left(1+\textsc{SNR}^{1-\gamma_{22}}\right)$\\
\item $(1-\lambda)S_{22}>\lambda\log\left(1+\textsc{SNR}^{1-\gamma_{22}}\right)\\
\;\;\;+(1-\lambda)\log\left(1+\frac{\textsc{SNR}^{1-\gamma_{22}}}{1+\textsc{SNR}^{b_2}}\right)$\\
\item $R_2>\log\left(1+\textsc{SNR}^{1-\gamma_{22}}\right)$.
\end{enumerate}
%\end{split}
%\end{equation}

We can notice that the first, fourth and fifth outage events are subsets of the last event, thus, can be
eliminated. Similarly, the second outage event is a subset of the sixth one.
\end{comment}

\bibliographystyle{IEEEbib}
\bibliography{MyLib}

\end{document}